\documentclass[fullpage]{article}
\usepackage{amsmath, amsthm, amssymb}
\usepackage{graphicx}
\usepackage{caption}
\usepackage{cite}
\usepackage{todonotes}
\usepackage[noend]{algorithmic}

\usepackage[margin=1in]{geometry} 
\usepackage{hyperref}	

\theoremstyle{definition}
\newtheorem{definition}{Definition}
\theoremstyle{plain}
\newtheorem{theorem}{Theorem}
\newtheorem{lemma}{Lemma}

\newtheorem{observation}{Observation}

\newtheorem{corollary}{Corollary}

\usepackage{tikz}

\begin{document}

\title{Independent Set on P$_k$-Free Graphs in Quasi-Polynomial Time}

\author{
Peter Gartland\thanks{University of California, Santa Barbara, USA. Emails: \texttt{petergartland@ucsb.edu}, \texttt{daniello@ucsb.edu}}
\and Daniel Lokshtanov\footnotemark[1]
}
\maketitle

\begin{abstract}
We present an algorithm that takes as input a graph $G$ with weights on the vertices, and computes a maximum weight independent set $S$ of $G$. If the input graph $G$ excludes a path $P_k$ on $k$ vertices as an induced subgraph, the algorithm runs in time $n^{O(k^2 \log^3 n)}$. Hence, for every fixed $k$ our algorithm runs in quasi-polynomial time. This resolves in the affirmative an open problem of~[Thomass\'{e}, SODA'20 invited presentation]. Previous to this work, polynomial time algorithms were only known for $P_4$-free graphs [Corneil et al., DAM'81], $P_5$-free graphs [Lokshtanov et al., SODA'14], and $P_6$-free graphs [Grzesik et al., SODA'19]. For larger values of $t$, only $2^{O(\sqrt{kn\log n})}$ time algorithms [Basc\'{o} et al., Algorithmica'19] and quasi-polynomial time approximation schemes [Chudnovsky et al., SODA'20] were known. Thus, our work is the first to offer conclusive evidence that {\sc Independent Set} on $P_k$-free graphs is not {\sf NP}-complete for any integer $k$.

Additionally we show that for every graph $H$, if there exists a quasi-polynomial time algorithm for {\sc Independent Set} on $C$-free graphs for every connected component $C$ of $H$, then there also exists a quasi-polynomial time algorithm for  {\sc Independent Set} on $H$-free graphs. This lifts our quasi-polynomial time algorithm to $T_k$-free graphs, where $T_k$ has one component that is a $P_k$, and $k-1$ components isomorphic to a {\em fork} (the unique $5$-vertex tree with a degree $3$ vertex).
\end{abstract}

\thispagestyle{empty}

\newpage

\setcounter{page}{1}

\section{Introduction}\label{sec:intro}
An {\em independent set} (also known as a {\em stable set}) in a graph $G$ is a vertex set $S$ such that no pair of distinct vertices in $S$ are adjacent in $G$. In the {\sc Independent Set} problem the input is a graph $G$ on $n$ vertices and integer $k$, the task is to determine whether $G$ contains an independent set $S$ of size at least $k$. {\sc Independent Set} is a well-studied and fundamental graph problem which is  {\sf NP}-complete~\cite{GareyJ79,Karp72} and intractable within most frameworks for coping with {\sf NP}-hardness. Indeed, {\sc Independent Set} was one of the very first problems to be shown to be {\sf NP}-hard to approximate~\cite{FeigeGLSS96,Zuckerman07}, one of the first intractable problems from the perspective of parameterized complexity~\cite{DF99}, one of the first problems to be shown not to have a $2^{o(n)}$ time algorithm assuming the Exponential Time Hypothesis (ETH)~\cite{ImpagliazzoPZ01}, and one of the very first problems whose hardness of parameterized approximation, assuming the Gap-ETH, was established~\cite{ChalermsookCKLM17}. 

With the above in mind, it is natural that a significant research effort has been devoted to identifying classes of input graphs for which the  {\sc Independent Set} problem is substantially easier than on general graphs. Of particular interest are the classes where  {\sc Independent Set} becomes polynomial time solvable. Most famously the problem becomes polynomial time solvable on Perfect graphs~\cite{GLS81}, other examples of polynomial time solvable cases include $k \times K_2$-free graphs~\cite{BY89} and graphs of bounded cliquewidth~\cite{CourcelleMR00}. 
For an extensive list, see~\cite{brandstadt1999graph} and the companion website~\cite{GCWebpage}. On the other hand the problem remains NP-complete even on planar graphs of maximum degree 3~\cite{GJ79}, unit disc graphs~\cite{CCJ90}, triangle-free graphs~\cite{Po74} and AT-free graphs~\cite{BKKM99}.

This paper fits in a long line of work to precisely classify the complexity of {\sc Independent Set} on all {\em hereditary} graph classes
defined by a single forbidden induced subgraph $H$ (and more generally, by a finite set ${\cal H}$ of forbidden induced subgraphs). A graph $G$ is said to be $H$-{\em free} if $G$ does not contain a copy of $H$ as an induced subgraph. For a set ${\cal H}$ of graphs, $G$ is ${\cal H}$-{\em free} if $G$ is $H$-free for all $H \in {\cal H}$. The ultimate goal of this research direction is to establish a dichotomy theorem that for every finite set ${\cal H}$ of graphs determines whether {\sc Independent Set} on ${\cal H}$-free graphs is polynomial time solvable, or {\sf NP}-complete\footnote{There is of course the possibility that {\sc Independent Set} on ${\cal H}$-free graphs has  {\sf NP}-intermediate complexity for some choice of ${\cal H}$. We believe this is unlikely, however that is pure speculation.}.

In 1982 Alekseev~\cite{Alekseev82} observed that {\sc Independent Set} remains {\sf NP}-complete on the class of ${\cal H}$-free graphs for every finite set ${\cal H}$ that does not include a graph $H$ whose every connected component is a path or a subdivision of the claw ($K_{1,3}$). Since then, no new {\sf NP}-completeness results for {\sc Independent Set} on ${\cal H}$-free graphs have been found for any other finite set ${\cal H}$. Thus, it is consistent with current knowledge that {\sc Independent Set} is polynomial time solvable on ${\cal H}$-free graphs for all other finite sets ${\cal H}$.
At the same time, progress on algorithms has been embarrassingly slow. The only connected graphs $H$ for which {\sf NP}-completness of {\sc Independent Set} does not follow from Alekseev's result are paths and subdivisions of the claw. Polynomial time algorithms for {\sc Independent Set} on claw-free graphs were found independently by Sbihi~\cite{Sbihi1980} and Minty~\cite{Minty80} in 1980. A polynomial time algorithm on {\em fork}-free graphs (a fork is a claw with one subdivided edge) was found by Alekseev~\cite{Alekseev04}. Subsequently, Lozin and Milanic~\cite{LozinM06} gave an algorithm for {\sc Weighted Independent Set} on fork-free graphs. 
For paths, {\sc Independent Set} on $P_4$-free graphs was shown to be polynomial time solvable by Corneil et al.~\cite{Corneil1981} in 1981. After a series of papers giving polynomial time algorithms for various subclasses of $P_5$-free graphs~\cite{BL03,BM03,GeLo03,LM05,Mo08,RaSc10}, in 2014 Lokshtanov et al.~\cite{LokshantovVV14} gave a polynomial time algorithm on $P_5$ free graphs. Two years later, Lokshtanov et al.~\cite{LokshtanovPL18} devised a quasi-polynomial time algorithm on $P_6$-free graphs, before Grzesik et al.~\cite{GrzesikKPP19} designed a polynomial time algorithm for $P_6$-free graphs in 2019. This summarizes the state-of-the-art for polynomial time solvability of {\sc Independent Set} on $H$-free graphs.

It appears that the currently known techniques are very far from being able to yield polynomial time algorithms for {\sc Independent Set} on $P_k$-free graphs for $k = 8$, let alone for all values of $k$. More concretely, the polynomial time algorithms for $P_5$-free graphs of Lokshtanov et al.~\cite{LokshantovVV14} and for $P_6$-free graphs of Grzesik et al.~\cite{GrzesikKPP19} are based on the same method. First, from a sample of two articles the complexity of applying this method seems to grow exponentially with $k$. Second, and more importantly, in a recent manuscript Grzesik et al.~\cite{corr/abs-2003-12345} show that a crucial component of this method fails completely on $P_k$-free graphs for $k \geq 8$.

The slow progress on polynomial time algorithms have prompted researchers to look for weaker forms of tractability of {\sc Independent Set} on $P_k$-free graphs. Bacs\'{o} et al.~\cite{BacsoLMPTL19} provided $2^{O(\sqrt{kn\log n})}$ time algorithms for {\sc Independent Set} on $P_k$-free graphs (see also~\cite{Brause17,GroenlandORSSS19}). Finally, Chudnovsky et al.~\cite{ChudnovskyPPT20} obtained quasi-polynomial time approximation schemes for $P_k$-free graphs for all $k$. In fact their result is much more general - they obtain quasi-polynomial time approximation schemes on ${\cal H}$-free graphs for all sets ${\cal H}$ for which {\sf NP}-hardness does not follow from Alekseev's~\cite{Alekseev82} observation. While the results above are general, they are consistent with {\sc Independent Set} being {\sf NP}-complete on all ${\cal H}$-free classes of graphs on which polynomial time algorithms are not already known. In this paper we obtain a quasi-polynomial time algorithm for {\sc Weighted Independent Set} on $P_k$-free graphs for every $k$. In particular we prove the following theorem.

\begin{theorem}\label{thm:main}
There exists an algorithm that given a graph $G$ and weight function $w : V(G) \rightarrow \mathbb{N}$ outputs the weight of a maximum weight independent set of $G$. If $G$ is $P_k$-free then the algorithm runs in $n^{O(k^2 \log^3 n)}$ time. 
\end{theorem}

Theorem~\ref{thm:main} implies that unless ${\sf NP} \subseteq {\sf QP}$, {\sc Independent Set} on $P_k$-free graphs is not {\sf NP}-complete for any $k$. This is the first conclusive evidence against {\sf NP}-completeness for any $k \geq 7$. The running time of the algorithm of Theorem~\ref{thm:main} matches that of Chudnovsky et al.~\cite{ChudnovskyPPT20}, but computes optimal solutions instead of $(1-\epsilon)$-approximate ones. It is also worth mentioning that our algorithm and analysis is substantially simpler than the quasi-polynomial time algorithm of Lokshtanov et al.~\cite{LokshtanovPL18} for the special case of $P_6$-free graphs.
We have been unsuccessful in generalizing Theorem~\ref{thm:main} to a quasi-polynomial time algorithms for $H$-free graphs where $H$ is a subdivision of a claw. However, the techniques used to prove Theorem~\ref{thm:main} can be used to show that such an algorithm would automatically generalize to all classes of ${\cal H}$-free graphs for which {\sf NP}-hardness is not already known. More concretely, for a graph $H$ let O$_H$ be an oracle that takes as input an $H$-free graph $G$ and outputs the weight of a maximum weight independent set in $G$. Further, let  ${\cal CC}(H)$ denote the set of connected components of $H$. Our second result is the following. 

\begin{theorem}\label{thm:main 2}
There exists an algorithm that given as input a graph $H$, a graph $G$, and weight function $w : V(G) \rightarrow \mathbb{N}$ as well as access to oracles $O(H_i)$ for all $H_i \in {\cal CC}(H)$, outputs the weight of a maximum weight independent set of $G$. If $G$ is $H$-free then the algorithm uses at most $n^{O(|H|^2|{\cal CC}(H)|\log^3 (n))}$ operations and oracle calls on induced subgraphs of $G$. 
\end{theorem}

Theorem~\ref{thm:main 2} has two immediate consequences. First, coupled with Theorem~\ref{thm:main} and the polynomial time algorithm for {\sc Weighted Independent Set} on fork-free graphs, Theorem~\ref{thm:main 2} yields a quasi-polynomial time algorithm for {\sc Weighted Independent Set} on $T_k$-free graphs, where $T_k$ is the graph with $k$ connected components the first of which is a $P_k$ and each of the remaining $k-1$ are isomorphic to a fork. Second, Theorem~\ref{thm:main 2} implies that if {\sc Weighted Independent Set} has a quasi-polynomial time algorithm on $H$-free graphs for every subdivided claw $H$, then {\sc Weighted Independent Set} also has a quasi-polynomial time algorithm on all ${\cal H}$-free classes of graphs, for finite sets ${\cal H}$, for which {\sf NP}-hardness does not follow from Alekseev's result. Or, stated more poetically, the buck stops at the (subdivided) claw.


\paragraph{Methods.}
The starting point for our algorithm is the $2^{O(\sqrt{n\log n})}$ time algorithm for $P_k$-free graphs of Basc\'{o} et al.~\cite{BacsoLMPTL19}. The algorithm of Basc\'{o} et al.~\cite{BacsoLMPTL19} is simple enough that we can give a quite detailed overview here.  It combines two methods - ``{\em degree reduction}'' and ``{\em balanced separation}''. 
 
The ``degree reduction'' approach can be summarised as follows. As long as the input graph $G$ contains a vertex $v$ of sufficiently high degree (degree $\geq d$) then {\em branch on} $v$. That is, find the best solution avoiding $v$ by a recursive call on $G - v$, and the best solution containing $v$ by adding $v$ to the solution obtained from a recursive call on $G - N[v]$. Output the best of these two solutions. A simple recurrence analysis shows that this reduces the problem to solving $2^{O(\frac{n \log n}{d})}$ instances in which no vertex has degree at least $d$. Basc\'{o} et al.~\cite{BacsoLMPTL19} set $d = \sqrt{n\log n}$ and obtain $2^{O(\sqrt{n \log n})}$ instances with maximum degree $\sqrt{n\log n}$.

The ``balanced separation'' technique is based on the classic ``Gy\'{a}rf\'{a}s path'' argument~\cite{gyarfas1987problems} for proving that P$_k$-free graphs are $\chi$-bounded. A simple lemma (Lemma 2 of~Basc\'{o} et al.~\cite{BacsoLMPTL19}), whose proof spans less than a page, shows that in every $P_k$ free graph $G$ there exists a vertex set $X_1$ of size at most $k-1$, such that every connected component of $G - N[X_1]$ has at most $n/2$ vertices. 
Basc\'{o} et al.~\cite{BacsoLMPTL19} apply this result to instances output by the degree reduction procedure above. In such instances, $|N[X_1]| \leq O(\sqrt{n\log n})$, assuming $k$ is a constant. Then, after guessing the intersection of the optimal solution with $N[X_1]$ (there are at most $2^{|N[X_1]|} \leq 2^{\sqrt{n\log n}}$ such guesses) the connected components of $G-N[X_1]$ become independent sub-instances of size at most $n/2$, on which the algorithm may be called recursively. 
Thus, solving a single instance on $n$ vertices reduces to solving $2^{O(\sqrt{n\log n})}$ instances on at most $n/2$ vertices. Analyzing the corresponding recurrence shows that the total running time of the algorithm is upper bounded by  $2^{O(\sqrt{n\log n})}$. 

If we wish to improve the running time from $2^{O(\sqrt{n\log n})}$ to quasi-polynomial, we may only apply degree reduction with $d = \Omega(\frac{n}{\log^{O(1)} n})$, and we can not afford to guess the intersection of the balanced separator $N[X_1]$ with an optimal solution. At this point we apply a slight generalization of degree reduction, to degree reduction relative to a vertex set $S$. Here we branch on vertices $v$ that have at least $d'$ neighbors in $S$ (the vertex $v$ itself does not have to be in $S$). A simple recursion analysis shows that this will reduce a single instance to $n^{|S|/d'}$ instances where every vertex has at most $d'$ neighbors in $S$. We apply degree reduction on the balanced separator $N[X_1]$ with $d' = |N[X_1]|/c$ for some constant $c$ (possibly depending on $k$). Thus, the initial degree reduction, followed by the degree reduction on $N[X_1]$, reduces the task of solving a single instance $G$ to that of solving the problem on $2^{\log^{O(1)} n}$ instances in which every vertex has degree at most $n / \log^{O(1)} n$ and furthermore has at most $|N[X_1]|/c$ neighbors in the set $N[X_1]$. Here we are working with induced subgraphs of the original graph $G$, so when we say $N[X_1]$ we really mean what remains of the set $N[X_1]$ (with the neighborhood taken in the graph $G$) in the subgraph of $G$ that is currently being considered. 

The route above is perhaps the most natural one to try to obtain a quasi-polynomial time algorithm. Indeed, it is also the engine in the quasi-polynomial time algorithm of Lokshtanov et al.~\cite{LokshtanovPL18} for $P_6$-free graphs. However it is not at all clear how to deal with the instances output by this degree reduction. For $P_6$-free graphs, Lokshtanov et al.~\cite{LokshtanovPL18} (essentially) show that if the balanced separator $N[X_1]$ is chosen very carefully, then the degree reduction procedure never gets stuck: as long as $N[X_1]$ is non-empty some vertex is a neighbor to a constant fraction of $N[X_1]$. Thus the algorithm will make quasi-polynomially many calls on instances where the balanced separator $N[X_1]$ has been reduced to the empty set, in which case each connected component of the graph is substantially smaller than the original graph. This leads to a recurrence that solves to quasi-polynomial time. We are not able to prove an analogous statement for $P_k$-free graphs for higher values of $k$, and so we are faced with the problem of how to deal with the degree-reduced instances described above.

The key insight of our algorithm is the following: {\em if we re-apply the ``Gy\'{a}rf\'{a}s path'' argument of  Basc\'{o} et al.~\cite{BacsoLMPTL19} on the degree-reduced instances to obtain a new balanced separator $N[X_2]$, then $N[X_2]$ can not have large intersection with $N[X_1]$}. This is because $N[X_2]$ is the neighbor set of a constant size set ($X_2$) and no vertex in the degree-reduced instance has many neighbors in $N[X_1]$. We now apply the degree reduction procedure again, this time on $N[X_2]$. If this reduction procedure completely reduces $X_1$ or $X_2$ to the empty set, or disconnects the graph into connected components so that the largest one has at most $0.9n$ vertices, then we have won, because the connected components of our instances are substantially smaller than on the original graph. If the procedure gets stuck then we obtain yet another balanced separator $X_3$, observe that $X_3$ has small intersection with $X_2$ and $X_1$, and do degree-reduction on $X_3$. And this keeps going, we keep adding new balanced separators into the mix until the degree-reduction procedure sufficiently disconnects the graph (i.e the largest connected component of the instances becomes sufficiently smaller than the original graph. The hard part of the analysis is to prove 
that the graph does become substantially disconnected by the time at most $O(\log n)$ separators have been added to the instance.

The actual final form of the algorithm is slightly different from what we describe above. Indeed, based on the ideas in the previous paragraph we can get an algorithm with running time $O(2^{n^{\epsilon}})$ for every $\epsilon > 0$, however to obtain quasi-polynomial time we need to be slightly more careful. The main difference is that we do not do degree reduction on each individual separator $N[X_i]$. Instead we define {\em level sets}. Level $i$ is the set of all vertices that appear in at least $i$ of the separators $N[X_1], \ldots, N[X_t]$ that we have constructed so far. We will maintain that throughout the course of the algorithm the size of level $i$ drops exponentially with $i$. Thus there will only be $O(\log n)$ levels, and we can afford to run degree reduction so that for each level, no vertex sees more than a $(\frac{1}{k\log n})^{O(1)}$ fraction of that level. Then, when we add a new separator, because it is the neighbor set of only a constant number of vertices, each level will increase by at most a factor of $1 + (\frac{1}{k\log n})^{O(1)}$ of the size of the previous level. Thus, such a process may continue to depth $(k\log n)^{O(1)}$ while maintaining the invariant that the size of the level $i$ drops exponentially with $i$. 

If recursion depth $\Omega(k\log n)$ is reached without sufficiently disconnecting the graph (i.e the largest connected component $C$ of the graph still has size at least $N/2$, where $N$ is the number of vertices in the original graph) this means that we have found $\Omega(k\log n)$ balanced separators for the graph such that no vertex is contained in more than $O(\log n)$ of them. A simple counting argument then shows that the average distance between pairs of vertices in the component $C$ has to be at least $\frac{k\log n}{\log n} \geq k$, contradicting that $G$ is $P_k$-free. This means that after recursion depth $O(k\log n)$, the graph has already been disconnected! At this point running the algorithm from scratch on each of the connected components yields at most $n$ instances of size at most $n/2$ which solves to quasi-polynomial time. 

Our algorithm for Theorem~\ref{thm:main 2} follows the same template as the algorithm for Theorem~\ref{thm:main}. The key difference is that instead of growing a sequence of balanced separators we grow a sequence of (neighborhoods of) induced copies in $G$ of connected components of $H$. Again the sequence has the property that the sets in the sequence do not overlap too much, so if we can grow the sequence to length $\Omega(|H|^{O(1)} \log n)$ then we can find an induced copy of $H$ in $G$.





\section{Preliminaries}\label{sec:prelim}
All graphs in this paper are assumed to be simple, undirected graphs. We denote the edge set of a graph $G$ by $E(G)$ and the vertex set of a graph by $V(G)$. If $v$ $\in$ $V(G)$, then we use $N[v]$ to denote the closed neighborhood of v, i.e. the set of all neighbors of $v$ together with $v$ itself. We use $N(v)$ to denote the set $N[v] - \{v$\}. If $X$ $\subseteq$ $V(G)$, then $N[X]$ = $\bigcup_{x \in X}N[x]$ and $N(X)$ = $N[X] - X$. We use ${\cal CC}(G)$ to denote the set of connected components of $G$. If $G_1$, $G_2$,..., $G_m$ are graphs, then we use $G_1 + G_2 +...+ G_m$ to denote the graph that that has vertex set $V(G_1) \cup V(G_2) \cup...\cup V(G_m)$, and edge set $E(G_1) \cup E(G_2) \cup...\cup E(G_m)$.

Given a weight function $w : V(G) \rightarrow \mathbb{N}$ the weight of a vertex set $S$ is defined as $w(S) = \sum_{v \in S} w(v)$. An {\em independent set} in $G$ is a vertex set $S$ such that no pair of vertices in $S$ have an edge between them. We define {\sf mwis}$(G)$ to be the weight of the maximum weight independent set in $G$. The $length$ of a path is the number of vertices in the path and we denote by $P_k$ the path of length $k$. If $X$ $\subseteq$ $V(G)$ then we will use $G(X)$ to denote the the graph induced by the vertex set $X$, and if it is clear from the context we will use $G-X$ to denote the graph $G(V(G)-X)$.

Given a positive number $k$ and a graph $G$, we call a set $S$ $\subset$ $V(G)$ a $c$-$balanced$ $separator$ if no connected component of $G - S$ has over $c$ vertices.
%
A {\em vertex multi-family} ${\cal F}$ is a collection of vertex sets that allows for multiple instances of its vertex sets. If ${\cal F}$ = \{$S_1, S_2, \ldots , S_n\}$ and $X$ is a set of vertices, then ${\cal F} - X$ is the vertex multi-family \{$S_1-X, S_2-X, \ldots , S_n-X\}$. For two vertex multi-families ${\cal A}$ and ${\cal B}$ their {\em union} is denoted by ${\cal A} \cup {\cal B}$ and is defined by the vertex multi-family that contains all elements of ${\cal A}$ and of ${\cal B}$. The multiplicity of an element $X$ in ${\cal A} \cup {\cal B}$ is its multiplicity in  ${\cal A}$ plus its multiplicity in ${\cal B}$.  We will use $\log(x)$ to denote the function $\lceil \log_2(x) \rceil$ throughout this paper.

\section{Quasi-Polynomial Time Algorithm for P$_k$-Free Graphs}\label{sec:p5alg}



In this section we prove Theorem~\ref{thm:main}. We will make use of the following balanced separator lemma from Basco et al.~\cite{BacsoLMPTL19}. 

\begin{lemma}\label{weak gyarfas}
\cite{BacsoLMPTL19} There exists an algorithm that given a graph $G$ runs in polynomial time and outputs an induced path $P$ in $G$ such that $N(V(P))$ is a $\frac{|V(G)|}{2}$-balanced separator of $G$.
\end{lemma}

We begin by proving a slight strengthening of Lemma~\ref{weak gyarfas}.

\begin{lemma}\label{strong gyarfas}
There exists an algorithm that takes as input a graph $G$, and a positive integer $i$ such that $2^i < |V(G)|$, runs in polynomial time and outputs a set $X$ such that $N[X]$ is a $\frac{|V(G)|}{2^i}$-balanced separator in $G$. Furthermore, if $G$ is $P_k$-free then $|X| \leq 2^{i+1} \cdot k$.
\end{lemma}

\begin{proof}
Let $G$ and $i$ be as in the statement of the lemma, the proof is by induction on $i$. For $i = 1$ the algorithm calls the algorithm of Lemma~\ref{weak gyarfas} and obtains a path $P$. It then returns $X = V(P)$. Lemma~\ref{weak gyarfas} guarantees that in this case $X$ satisfies the statement of this lemma. For $i > 1$ the algorithm first calls itself recursively on the input $(G, i-1)$ and obtains a set $X'$ such that $N[X']$ is a $\frac{|V(G)|}{2^{i-1}}$-balanced separator in $G$, and furthermore, if $G$ is $P_k$-free then $|X'| \leq 2^{i} \cdot k$. For each connected component $C_j$ of $G-N[X']$ such that $|V(C_j)| > \frac{|V(G)|}{2^i}$ the algorithm calls itself recursively on $(C_j, 1)$ and obtains a set $X_j$ such that $N[X_j]$ is a $\frac{|V(C_j)|}{2}$-balanced separator of $C_j$. If $G$ is $P_k$-free it holds that $|X_j| \leq k$. The algorithms sets $X$ as  
$X = X' \cup (\bigcup_{j} X_j)$
where the union is taken over all $j$ such that  $|V(C_j)| > \frac{|V(G)|}{2^i}$. Clearly the construction of $X$ ensures that $N[X]$ is a $\frac{|V(G)|}{2^i}$-balanced separator of $G$. Furthermore if $G$ is $P_k$-free then  $|X| \leq |X'| + t \cdot k$ where $t$ is the number of connected components of $G-X'$ whose size is more than $\frac{|V(G)|}{2^i}$. Since these components are disjoint we have that $t \leq 2^i$. Therefore $|X| \leq 2^i \cdot k + 2^i \cdot k = 2^{i+1} \cdot k$ as claimed. 

To see that the run time is polynomial it suffices to show the number of times the algorithm makes a call to the algorithm of Lemma \ref{weak gyarfas} is polynomial. To see this polynomial bound, note that on input ($G, i$) the algorithm makes at most $2^i$ calls to the algorithm of Lemma \ref{weak gyarfas} plus the number calls it makes to the algorithm of Lemma \ref{weak gyarfas} on input ($G, i-1$). Since 2$^{i} \leq |V(G)| = n$, the recurrence shows the algorithm makes at most $\Sigma_{i=0}^{\log(n)} n/2^i$ = $O(n)$ calls to the algorithm of Lemma \ref{weak gyarfas}.
\end{proof}

To describe the algorithm of Theorem~\ref{thm:main} we first need to define the notion of {\em level sets} relative to a vertex multi-family ${\cal F}$. 

\begin{definition}\label{level sets}
Given a graph $G$ and a vertex multi-family ${\cal F}$ consisting of vertex sets of $G$, for positive integers $i$, the {\em $i$\textsuperscript{th} level relative to}  ${\cal F}$ is denoted by $L({\cal F}, i)$ and defined as follows
$$L({\cal F}, i) = \{v \in V(G) ~:~ |\{S \in {\cal F} ~:~ v \in S\}| \geq i\}$$
\end{definition}
In other words $L({\cal F}, i)$ is a vertex set containing all vertices of $G$ that are contained in at least $i$ sets in ${\cal F}$. Our algorithm will also make use of a number $N$, this number will be approximately equal to the number of vertices in the input graph $G$. 

\begin{definition}\label{branchableVertex}
The $i$\textsuperscript{th} {\em branch threshold} is denoted by $\Delta_i$ and is defined as $\Delta_i = N/2^i$. Given a multi-family ${\cal F}$, a vertex $v \in V(G)$ is a {\em branchable vertex} if there exists an $i \geq 1$ such that $|N[v] \cap L({\cal F}, i)| \geq \Delta_i$.
\end{definition}

In the following $G$ is always a graph, $w$ is a weight function $w : V(G) \rightarrow \mathbb{N}$, $N$ is an integer, and ${\cal F}$ is a multi-family of subsets of $V(G)$. We now describe the main subroutine ALG$_1$ in the algorithm of Theorem~\ref{thm:main}. The algorithm takes as input $G$, $w$, $N$ and ${\cal F}$ and (as we will prove) outputs the weight of a maximum weight independent set in $G$. The algorithm of Theorem~\ref{thm:main} will call ALG$_1$ with parameters G, $N=|V(G)|$, w, and ${\cal F} = \emptyset$.
ALG$_1$ is a recursive branching algorithm with only four rules. First, if $G$ has at most one vertex, then return $V(G)$. Second, if the largest component of $G$ has at most $|N|/2$ vertices then solve the problem recursively on each component and return the sum. Third, if there exists a branchable vertex $v$, then branch on $v$ (i.e solve the problem with $v$ forced in to the independent set, and $v$ forced out). Finally, if none of the previous rules apply then {\em add} a new balanced separator $X$ (obtained by Lemma~\ref{strong gyarfas}) to ${\cal F}$. In other words, make a recursive call on the instance $(G, w, N, {\cal F} \cup \{N[X]\})$.

ALG$_1$ is very similar to well known exact exponential time branching algorithms for {\sc Independent Set}~\cite{FominK10}. The key differences are that we use the multi-family ${\cal F}$ of balanced separators to guide which vertex to branch on, that when no rules apply we add a separator to the family ${\cal F}$ (at a glance this appears to make no progress at all, but it increases the size of the level sets, making more vertices branchable), and that we wait with recursing on connected components until the size of the largest component becomes less than $N/2$ (this is primarily to simplify the analysis). 

\medskip
\noindent \textbf{ALG$_1$}
\begin{algorithmic}[1]
\STATE \textbf{Input:} $G$, $w$, $N$, ${\cal F}$.
\STATE \textbf{Output:} {\sf mwis}$(G)$.

\IF{ $|V(G)| \leq 1$}  
\RETURN $w(V(G))$
\ELSIF{($\max_{C \in {\cal CC}(G)} |V(C)| )\leq N/2$} 
\RETURN $\sum_{C \in {\cal CC}(G)} \mbox{ALG}_1(C, w, |V(C)|, \emptyset)$
\ELSIF{exists branchable vertex $v$}
\RETURN $\max\left(\mbox{ALG}_1(G - v, w, N, {\cal F} - \{v\}), \mbox{ALG}_1(G - N[v], w, N, {\cal F} - N[v]) + w(v) \right)$
\ENDIF
\STATE \textbf{obtain} $X$ by applying Lemma~\ref{strong gyarfas} on $G$ with $i = 2$
\RETURN ALG$_1(G, w, N, {\cal F} \cup \{N[X]\})$
\end{algorithmic}
\medskip

We will distinguish between the three different kinds of recursive calls that ALG$_1$ can make. If the {\bf else if} condition on line 5 holds, then the algorithm makes the recursive calls on line 6. In this case we say that ALG$_1$ {\em recurses on connected components}. If the {\bf else if} condition on line 7 holds, then the algorithm makes the recursive calls on line 8. In this case we say that ALG$_1$ {\em branches on a branchable vertex}. Otherwise the algorithm makes the recursive call on line 10. In this case we say that ALG$_1$ {\em adds a balanced separator}. We will frequently need to refer to parts of the execution of the algorithm. For disambiguation, we collect the terminology here. 

%

An {\em instance} is a four-tuple $(G, w, N, {\cal F})$.
A {\em run} of the algorithm refers to the entire execution of the algorithm on an instance.
A {\em call} $\mbox{ALG}_1(G, w, N, {\cal F})$ refers to the computation done in the root node of the recursion tree of the run  $\mbox{ALG}_1(G, w, N, {\cal F})$. We remark that parameters $G, w, N$, and ${\cal F}$ never change during the call $\mbox{ALG}_1$($G, w, N, {\cal F})$.
When a run or a call ALG$_1(G, w, N, {\cal F})$ recursively calls ALG$_1$ on the instance $(G', w, N', {\cal F}')$ we say the run or the call {\em executes} a run or a call on $(G', w, N', {\cal F}')$.  This will sometimes be referred to as {\em makes a recursive call} $\mbox{ALG}_1(G', w, N', {\cal F}')$.
A run of $\mbox{ALG}_1(G,w,N,{\cal F})$ is called a $k$-$fair$ $run$ if $G$ is a $P_k$-free graph, $N = |V(G)|$, ${\cal F} = \emptyset$, and $w$ is a weight function. A call $\mbox{ALG}_1(G, w, N, {\cal F})$ is called a $k$-$fair$ $call$ if it is executed during the course of a $k$-fair run. An instance $(G, w, N, {\cal F})$ is called a {\em $k$-fair} instance if $\mbox{ALG}_1(G, w, N, {\cal F})$ is a k-fair call. 

 
\begin{lemma}\label{termination}
ALG$_1(G,w,N,{\cal F})$ terminates on every input. 
\end{lemma}

\begin{proof}
Consider a run of ALG$_1$ with initial input $G, w, N,$ and ${\cal F}$. Whenever the algorithm makes a recursive call ALG$_1(G',w,N',{\cal F}')$ we have that $|V(G')| \leq |V(G)|$ and $N' \leq N$. Furthermore, whenever the algorithm recurses on connected components or branches on a branchable vertex, then it executes ALG$_1(G',w,N',{\cal F}')$ with either $|V(G')| < |V(G)|$ or $N' < N$. Finally, ALG$_1$ cannot add a balanced separator for over $|V(G)| \cdot \log(N)$ successive recursive calls since then a call ALG$_1(G,w,N,{\cal F}'')$ with ${\cal F''} = |V(G)| \cdot \log(N)$ would add a balanced separator. However, since each new balanced separator must be non-empty (since otherwise the algorithm would have recursed on connected components) we have that $L({\cal F''}, \log(N)) \neq \emptyset$, and so the call ALG$_1(G,w,N,{\cal F}'')$ would branch on a vertex. This contradicts that the call added a balanced separator, and proves that ALG$_1$ cannot add a balanced separator for over $|V(G)| \cdot \log(N)$ successive recursive calls. It follows by induction on $|V(G)| + N$ that ALG$_1$ always terminates.
\end{proof}




\begin{lemma}\label{mis return}
A run ALG$_1(G, w, N, {\cal F})$ always returns the weight of a maximum weight independent set of $G$ under the weight function $w$.
\end{lemma}
\begin{proof}
Consider a run of ALG$_1$ with initial input $G, w, N,$ and ${\cal F}$. It is clear from the algorithm that if each run ALG$_1$($G', w, N', {\cal F}$) that is executed by the call ALG$_1$($G, w, N, {\cal F}$) returns the weight of a maximum weight independent set of $G'$ with weight function $w$, then so would the run ALG$_1$($G, w, N,$ ${\cal F}$). By Lemma \ref{termination} the height of the recursion tree is bounded, and the result is trivially true for the base case of $|V(G)| \leq 1$, so the result follows by induction on the height of the recursion tree of the run ALG$_1$($G, w, N, {\cal F}$). 
\end{proof}





We have now proved that  ALG$_1$ always terminates and that it always outputs the correct answer. The remainder of the section is devoted to the running time analysis.
We will now prove some lemmas to help us bound the run time of ALG$_1$ on $k$-fair runs. First, in Observation~\ref{fair run balanced separator} we will prove that ${\cal F}$ remains a multi-family of balanced separators of $G$ throughout the execution of the algorithm. In Observation~\ref{log(N) + 1 is empty} we will show that no vertex appears in many (more than $\log N$) sets in ${\cal F}$. This will ensure that ${\cal F}$ can never grow too large, because, as we will show in Lemma \ref{F size}, a connected $P_k$-free graph can not contain a large fractional packing of balanced separators. 

\begin{observation}\label{fair run balanced separator}
Let ($G$, $w$, $N$, ${\cal F}$) be a $k$-fair instance. 
Then every set  $S \in {\cal F}$ is a $\frac{N}{4}$-balanced separator of $G$.
\end{observation}

\begin{proof}
Consider a $k$-fair instance $(G, w, N, {\cal F})$. If ${\cal F} = \emptyset$ then the result is trivially true, so assume ${\cal F} \neq \emptyset$. It follows ALG$_1$ executes ALG$_1(G, w, N, {\cal F})$ during a $k$-fair call ALG$_1(G', w, N, {\cal F'})$ by branching on a branchable vertex or ALG$_1$ executes ALG$_1(G, w, N, {\cal F})$ during a $k$-fair call ALG$_1(G, w, N, {\cal F''})$ by adding a balanced separator, $X$. In the first case, if all sets of ${\cal F'}$ are $\frac{N}{4}$-balanced separators for $G'$, then since $G$ = $G'-S$ for some vertex set $S$, and ${\cal F}$ = ${\cal F'}-S$, all sets of ${\cal F}$ are $\frac{N}{4}$-balanced separators for $G$. In the second case, $X$ is generated in such a way that it is guaranteed to be an $\frac{N}{4}$-balanced separator for $G$, so if all sets of ${\cal F''}$ are $\frac{N}{4}$-balanced separators for $G$, then all sets of ${\cal F}$ are $\frac{N}{4}$-balanced separators for $G$. The statement of the observation now follows by induction on the depth of the call ALG$_1(G, w, N, {\cal F})$ in the recursion tree.
\end{proof}



\begin{observation}\label{log(N) + 1 is empty}
For every $k$-fair instance $(G, w, N, {\cal F})$, we have that $L({\cal F},\log(N)+1) = \emptyset$. 
\end{observation}

\begin{proof}
Consider a $k$-fair call ALG$_1(G, w, N, {\cal F})$. We will prove the statement by induction on the depth the call ALG$_1(G, w, N, {\cal F})$ in the recursion tree of a run ALG$_1(G^*,w,|V(G^*)|,\emptyset)$ which executes ALG$_1(G, w, N, {\cal F})$.

If ${\cal F} = \emptyset$ then the result is trivially true. Suppose now that ${\cal F} \neq \emptyset$, it follows ALG$_1$ executes ALG$_1(G, w, N, {\cal F})$ during a $k$-fair call ALG$_1(G', w, N, {\cal F'})$ by branching on a branchable vertex 
or by adding a balanced separator $X$. In the first case ${\cal F}$ =  ${\cal F'}-S$ for some vertex set $S$. By the induction hypothesis $L({\cal F'},\log(N)+1) = \emptyset$ and hence $L({\cal F},\log(N)+1) = \emptyset$. 
In the second case, ALG$_1(G, w, N, {\cal F'})$ does not branch on a branchable vertex, so we have that $L({\cal F'},\log(N)) = \emptyset$ since every vertex in $L({\cal F'},\log(N))$ is branchable. It follows that $L({\cal F},\log(N)+1) =  L({\cal F'} \cup X,\log(N)+1) = \emptyset$.
\end{proof}

\begin{lemma}\label{F size}
For every $k$-fair instance $(G, w, N, {\cal F})$ it holds that $|{\cal F}| \leq 10k \cdot \log(N)$.
\end{lemma}

\begin{proof}
Consider a $k$-fair instance $(G, w, N, {\cal F})$. We will prove the result by induction on the depth of the call ALG$_1(G, w, N, {\cal F})$ in the recursion tree of a run ALG$_1(G^*,w,|V(G^*)|,\emptyset)$ which executes ALG$_1(G, w, N, {\cal F})$. 

In the base case ${\cal F} = \emptyset$, and the claim of the lemma holds trivially, so assume ${\cal F} \neq \emptyset$. Thus the call ALG$_1(G, w, N, {\cal F})$ is executed by a $k$-fair call ALG$_1(G', w, N', {\cal F}')$. By the induction hypothesis $|{\cal F}'| \leq 10k \cdot \log(N)$  ($N = N'$ since ${\cal F} \neq \emptyset$). Thus, unless ALG$_1(G', w, N', {\cal F}')$ recurses by adding a balanced separator we have that $|{\cal F}| \leq 10k \cdot \log(N)$ as well. So assume that ALG$_1(G', w, N', {\cal F}')$ adds a balanced separator $X$ and that therefore $G'=G$, $N'=N$ and ${\cal F} = {\cal F}' \cup \{X\}$. We prove that $|{\cal F}'| < 10k \cdot \log(N)$, then the result follows since  $|{\cal F}| = |{\cal F}'| + 1$.



Suppose for contradiction that $|{\cal F}'| \geq 10k \cdot \log(N)$, we will now produce an induced path of length $k$ in $G$, contradicting that $G$ is $P_k$-free.
The call $\mbox{ALG}_1(G', w, N', {\cal F'}) = \mbox{ALG}_1(G, w, N, {\cal F'})$ added a balanced separator, and so the size of the largest connected component, $C$, in $G$ is greater than $\frac{N}{2}$. This, together with Observation \ref{fair run balanced separator} then gives that every set $S \in {\cal F}$ is a $\frac{|V(C)|}{2}$-balanced separator for $C$. Consider the following random process. Uniformly at random, select vertices $x$ and $y$ in $C$. For all $S \in {\cal F}$, let $X_S$ denote the random variable that is 1 if $x$ and $y$ are not in the same connected component of $C-S$ and $0$ otherwise. Since $S$ is a $\frac{|V(C)|}{2}$-balanced separator for $C$, the probability that $x$ and $y$ are in the same connected component of $C-S$ is at most $\frac{1}{2}$. Thus $X_S = 1$ with probability at least $\frac{1}{2}$. We denote by ${\cal F}_{x,y}$ all sets $S \in {\cal F}$ such that $x$ and $y$ are not in the same component of $C-S$, again including multiplicity. By linearity of expectation we have that
$$E[|{\cal F}_{x,y}|] = \sum_{S \in {\cal F}} E[X_S] \geq |{\cal F}|/2 > 5k \cdot \log(N).$$ 

It follows there exists vertices $a$ and $b$ in $C$ such that $|{\cal F}_{a,b}| >  5k \cdot \log(N)$. Let $P$ be a shortest path connecting $a$ and $b$ in $C$. Since $G$ is $P_k$-free, $P$ has at most $k-1$ vertices. By Observation~\ref{log(N) + 1 is empty}, each of these vertices is contained in at most $\log(N)$ sets in ${\cal F}_{a,b}$. But then there exists a set $S \in {\cal F}_{a,b}$ disjoint from $V(P)$ contradicting that $a$ and $b$ are not in the same component of $C - S$.
%
\end{proof}



\medskip



The following observation shows that the level sets do not grow a lot in each successive recursive call, and that they therefore never get very large. Note in particular that the size of level set $i$ drops exponentially with $i$.

\begin{observation}\label{level set size}
For every $k$-fair call ALG$_1(G, w, N, {\cal F})$ that adds a balanced separator $X$ and every $i$,
$$|L({\cal F}\cup X,i)| \leq \Delta_{i-1} \cdot 8k + |L({\cal F},i)|.$$ 
Furthermore, for every $k$-fair instance $(G', w, N', {\cal F}')$,
$$|L({\cal F}',i)| \leq \Delta_{i-1} \cdot 8k \cdot |{\cal F}'|.$$
\end{observation}

\begin{proof}
Consider a $k$-fair call ALG$_1(G, w, N, {\cal F})$ that adds a balanced separator $X$. Let $X_j$ denote the set of vertices in $L({\cal F},j) \cap X$, then we can see that $|L({\cal F}\cup \{X\},j)| \leq L({\cal F},j) + |X_{j-1}|.$ Since the call ALG$_1(G, w, N, {\cal F})$ adds a balanced separator, $X$, there are no branchable vertices. So, we have that for all $v \in G$, $|N[v] \cap L({\cal F}, j)| \leq \Delta_j$. Furthermore, by Lemma \ref{strong gyarfas}, since $X$ is generated as an $\frac{N}{4}$-balanced separator and therefore a $\frac{|G|}{4}$-balanced separator for $G$, $X$ is the neighborhood of at most $8k$ vertices, hence $|X_{j-1}|$ $\leq$ $\Delta_{j-1} \cdot 8k$ and the result $|L({\cal F}\cup \{X\},i)| \leq \Delta_{i-1} \cdot 8k + |L({\cal F},i)|$ follows.

The second statement follows by combining induction, the first part of this Observation, and the fact that if the call ALG$_1(G, w, N, {\cal F})$ executes ALG$_1(G', w, N', {\cal F}')$, then $|{\cal F}| < |{\cal F}'|$ if and only if the call ALG$_1(G, w, N, {\cal F})$ adds a balanced separator.


\end{proof}

For $k$-fair instances $(G, w, N,{\cal F})$ we define a measure:
\begin{equation*}
 \label{eq2}
\begin{split}
\mu_k(G, w, N,{\cal F})  &= 400k^2 \cdot \log^2(N) \cdot (N+|V(G)|) + \sum_i \left(|L({\cal F},i)| \cdot \frac{N}{\Delta_{i-1}}\right)\\
& \quad + 16k \cdot N \cdot \log(N) \cdot (10k \cdot \log(N)-|{\cal F}|)\\
\end{split}
\end{equation*}
%
%

If $(G, w, N,{\cal F})$ is not a $k$-fair instance, then $\mu_k(G, w, N,{\cal F})$ is undefined. Note that $\mu_k(G, w, N,{\cal F})$ must always be an integer, and that it is independent of the weight function $w$. We will say that two instances $(G, w, N,{\cal F})$ and $(G', w', N',{\cal F}')$ are {\em essentially different} if $G' \neq G$, $N' \neq N$ or ${\cal F}' \neq {\cal F}$. 

\begin{lemma}\label{finite measure}
For every positive integer $\mu$, the number of essentially different $k$-fair instances $(G, w, N,{\cal F})$ such that $\mu_k(G, w, N,{\cal F})$ = $\mu$ is finite. In addition, for every $k$-fair instance it holds that  $\mu_k(G, w, N,{\cal F}) \geq 0$.
\end{lemma}

\begin{proof}
Consider a $k$-fair instance $(G, w, N, {\cal F})$ with $\mu_k(G, w, N,{\cal F})$ = $\mu$. Clearly, there are only a finite number of such instances with $N = 1$. We will show that if $N \geq 2$ then $|V(G)| \leq \mu$. If $|V(G)| \leq \mu$ then $|G|$, $|N|$, and $|{\cal F}|$ are all bounded in terms of $\mu$, and the first part of the lemma follows.

By Lemma \ref{F size} we have that $|{\cal F}|$ is at most 10$k \cdot \log(N)$. It follows that the terms $400k^2 \cdot \log^2(N) \cdot (N+|V(G)|)$, $\Sigma_i(|L({\cal F},i)| \cdot \frac{N}{\Delta_{i-1}})$, and $16k \cdot N \cdot \log(N) \cdot (10k \cdot \log(N)-|{\cal F}|)$ are all non negative. Hence $\mu_k(G, w, N,{\cal F})$ $\geq$ $|V(G)|$. This also proves that $\mu_k(G, w, N,{\cal F})$ $\geq$ 0.
\end{proof}

\begin{lemma}\label{max measure size}
For every $k$-fair instance $(G, w, N,{\cal F})$ it holds that $\mu_k(G, w, N,{\cal F}) \leq 1050k^2 \cdot N \cdot \log^2(N)$
\end{lemma}

\begin{proof}
Consider a $k$-fair instance $(G, w, N,{\cal F})$. By Observation \ref{level set size} and Lemma \ref{F size}, we have that  $|L({\cal F},i)| \cdot \frac{N}{\Delta_{i-1}} < 8k \cdot N \cdot |{\cal F}| < 80k^2 \cdot N \cdot \log(N)$. Therefore, $\Sigma_i(|L({\cal F},i)| \cdot \frac{N}{\Delta_{i-1}})$ $<  80k^2 \cdot N \cdot \log^2(N)$. Also, since $N \geq |V(G)|$, we can see that 
%
\begin{equation*}
 \label{eq2}
\begin{split}
\mu_k(G, w, N,{\cal F}) &= 400k^2 \cdot \log^2(N) \cdot (N+|V(G)|) + \Sigma_i(|L({\cal F},i)| \cdot \frac{N}{\Delta_{i-1}}) \\
& \quad + 16k \cdot N \cdot \log(N) \cdot (10k \cdot \log(N)-|{\cal F}|) \\
& < 800k^2 \cdot N \cdot \log^2(N) +  80k^2 \cdot N \cdot \log^2(N) + 160k^2 \cdot N \cdot \log^2(N)\\ 
& < 1050k^2 \cdot N \cdot \log^2(N)
\end{split}
\end{equation*}

\end{proof}

We define $T_k(G, w, N, {\cal F})$ to be the running time of a $k$-fair run of ALG$_1$ starting with the inputs $(G, w, N, {\cal F})$. We also define $$T_k(\mu) = max_{G, N,{\cal F} ~:~ \mu_k(G, w, N,{\cal F}) \leq \mu} \mbox{T}_k(G, w, N,{\cal F}).$$

When we analyze run time we assume that arithmetic (addition, subtraction, comparisons) on weights of vertices and vertex sets is constant time. Thus, both the running time of $ALG_1$ and the measure of an instance  $(G, w, N, {\cal F})$ are independent of the weight function $w$. Thus, by Lemma \ref{finite measure}, $T_k(\mu)$ is well defined. 



\begin{lemma}\label{T recurrence 1}
$T_k(\mu)$ satisfies the following recurrence:
%
\begin{equation*}
T_k(\mu) \leq \mu^{O(1)} + max
\begin{cases}
\mu T_k(.95\mu) \\
T_k(\mu - 1) + T_k(\mu[1-1/(2100k^2 \cdot \log^2(\mu))])\\
T_k(\mu[1-1/(200k \cdot \log(\mu))])\\
\end{cases}
\end{equation*}
\end{lemma}

\begin{proof}
Let $(G, w, N,{\cal F})$ be a $k$-fair instance such that $\mu_k(G, w, N,\emptyset) = \mu$ and $T_k(\mu)$ is the run time of ALG$_1(G, w, N,{\cal F})$. If the call ALG$_1(G, w, N,{\cal F})$ recurses on connected components, then it makes at most $|V(G)|$ recursive calls on instances of the form $(G', w, N',\emptyset)$, where $|V(G')| \leq |V(G)|$ and $N' \leq \frac{N}{2}$. It follows that for each of these recursive calls we have  

\begin{align*}
 \label{eq1}
\mu_k(G', w, N',\emptyset) & = 400k^2 \cdot \log^2(N') \cdot (N'+|V(G')|) + 160k^2 \cdot N' \cdot \log^2(N') \\
& \leq 400k^2 \cdot \log^2(N) \cdot (\frac{N}{2}+|V(G)|) + 80k^2 \cdot N \cdot \log^2(N) \\
& \leq 400k^2 \cdot \log^2(N) \cdot (N + |V(G)|) - 100k^2 \cdot N \cdot \log^2(N) \\ 
& \leq \mu - 100k^2 \cdot N \cdot \log^2(N) \\
& \leq .95\mu \mbox{~~~~~(by\ Lemma~\ref{max measure size})}
\end{align*}

Therefore, if the instance ALG$_1(G, w, N,{\cal F})$ recurses on connected components, we must have that $T_k(\mu) \leq |V(G)| \cdot T_k(.95\mu) \leq \mu \cdot T_k(.95\mu)$.

If the call ALG$_1(G, w, N,{\cal F})$ branches on a branchable vertex, $v$, then it makes two recursive calls, one execution ALG$_1(G-\{v\}, w, N, {\cal F}-\{v\})$, where the instance $(G-\{v\}, w, N, {\cal F}-\{v\})$ has measure $\mu_k(G-\{v\}, w, N,{\cal F}-\{v\}) \leq \mu - 1$, and the other execution is ALG$_1(G-N[v], w, N, {\cal F}-N[v])$. Note that for a branchable vertex, $v$, we have that 
$$\sum_i(|L({\cal F}-N[v],i)| \cdot \frac{N}{\Delta_{i-1}} \leq \sum_i(|L({\cal F},i)| \cdot \frac{N}{\Delta_{i-1}}) - \frac{N}{2},$$
since for at least one level $i$ we have that $|N[v] \cap L({\cal F},i)| \geq \Delta_i$ and $\frac{\Delta_i}{\Delta_{i-1}}$ = 1/2. It follows that 
%
%
\begin{align*}
\begin{split}
\mu_k(G-N[v], w, N,{\cal F}-N[v])  & = 400k^2 \cdot \log^2(N) \cdot (N+|V(G)-N[V]|) + \sum_i\left(|L({\cal F}-N[v],i)| \cdot \frac{N}{\Delta_{i-1}}\right)  \\
&  \quad + 16k \cdot N \cdot \log(N) \cdot (10k \cdot \log(N)-|{\cal F}|)  \\
& \leq 400k^2 \log^2(N) \cdot (N+|V(G)|) + \sum_i \left(|L({\cal F},i)| \cdot \frac{N}{\Delta_{i-1}}\right) \\
& \quad + 16k \cdot N \cdot \log(N) \cdot (10k \cdot \log(N)-|{\cal F}|) - \frac{N}{2} \\
& \leq \mu - \frac{N}{2}  \\ 
& \leq \mu\left(1-\frac{1}{2100k^2 \cdot \log^2(N)}\right)  \mbox{~~~~~(by\ Lemma\ \ref{max measure size})} \\
& \leq \mu\left(1-\frac{1}{2100k^2 \cdot \log^2(\mu)}\right) 
\end{split}
\end{align*}
Therefore, if the call ALG$_1(G, w, N,{\cal F})$ branches on a branchable vertex, then we have that $T_k(\mu)$ $\leq$ $T_k(\mu - 1) + T_k(\mu[1-\frac{1}{2100k^2 \cdot \log^2(\mu)}])$.

Finally, if the call ALG$_1(G, w, N,{\cal F})$ adds a balanced separator, $X$, then it makes a single recursive call ALG$_1(G, w, N,{\cal F}\cup X)$. By Observation \ref{level set size} and Lemma \ref{max measure size} we obtain the following.
$$\mu_k(G, w, N,{\cal F}\cup X) < \mu + 8k \cdot N \cdot \log(n) - 16k \cdot N \cdot \log(N) < \mu\left([1-\frac{1}{200k \cdot \log(\mu)}]\right)$$ 
Therefore, if the call ALG$_1(G, w, N,{\cal F})$ adds a balanced separator, then $T_k(\mu)$ $\leq$ $T_k(\mu[1-\frac{1}{200k \cdot \log(\mu)}])$.

The result now follows from the observation that ALG$_1(G, w, N,{\cal F})$ only does $|V(G)|^{O(1)}$ = $\mu^{O(1)}$  work in any given call and always recurses on connected components, branches on a branchable vertex, adds a balanced separator, or returns without making further recursive calls.
\end{proof}

Since $T_k(\mu)$ is a non negative, non decreasing function, by adding the three possibilities in the $\max$ of  Lemma \ref{T recurrence 1} we immediately obtain the following simplified recurrence. 


\begin{corollary}\label{T bound}
$T_k(\mu) \leq \mu^{O(1)} + \mu T_k(.95\mu) + T_k(\mu - 1) + T_k(\mu[1-\frac{1}{2100k^2 \cdot \log^2(\mu)}]) + T_k(\mu[1-\frac{1}{200k \cdot \log(\mu)}]) < T_k(\mu - 1) + \mu^{O(1)} +  3\mu \cdot T_k(\mu[1-\frac{1}{2100k^2 \cdot \log^2(\mu)}])$
\end{corollary}


\begin{lemma}\label{T asymptotics}
$T_k(\mu) = \mu^{O(k^2 \cdot \log^3(\mu))}$
\end{lemma}

\begin{proof}
The proof is by induction on $\mu$. The base case is established by Lemma \ref{finite measure}. By Corollary \ref{T bound} we have the inequality $T_k(\mu) \leq T_k(\mu - 1) +  \mu^{O(1)} + 3\mu T_k(\mu[1-\frac{1}{2100k^2 \cdot \log^2(\mu)}])$ and repeatedly applying the inequality to the first term on the right hand side, gives $T_k(\mu) \leq \mu^{O(1)} + 3\mu^2 \cdot T_k(\mu[1-\frac{1}{2100k^2 \cdot \log^2(\mu)}])$. By the inductive hypothesis then, there is some $c$ such that

\begin{equation*}
 \label{eq1}
\begin{split}
T_k(\mu) & \leq \mu^{O(1)} + 3\mu^2 \cdot (\mu[1-\frac{1}{2100k^2 \cdot \log^2(\mu)}])^{ck^2 \cdot \log^3(\mu)}\\
&  = \mu^{O(1)} + 3\mu^2 \cdot \mu^{ck^2 \cdot \log^3(\mu)} \cdot [1-\frac{1}{2100k^2 \cdot \log^2(\mu)}]^{ck^2 \cdot \log^3(\mu)}\\
& \leq \mu^{O(1)} + 3\mu^2 \cdot \mu^{ck^2 \cdot \log^3(\mu)} \cdot e^{-\frac{ck^2 \cdot \log^3(\mu)}{2100k^2 \cdot \log^2(\mu)}} \mbox{~~~~ (~since~} 1-x \leq e^{-x} \mbox{~)}\\
& \leq  \mu^{O(1)} + 3\mu^2 \cdot \mu^{ck^2 \cdot \log^3(\mu)} \cdot e^{-\frac{c\log(\mu)}{2100}} \\
& \leq \mu^{ck^2 \cdot \log^3(\mu)} \mbox{~~~~~~( for sufficiently large $c$ )}\\ 
\end{split}
\end{equation*}


\end{proof}

We are now ready to prove Theorem \ref{thm:main}.

\begin{proof}[Proof of Theorem~\ref{thm:main}]
The algorithm returns the answer of  ALG$_1(G, |V(G)|, w, \emptyset)$. By Lemma~\ref{termination} ALG$_1$ terminates, by  Lemma~\ref{mis return} ALG$_1$ returns return the weight of a maximum weighted independent set.
For the running time, observe that $(G, w, N, \emptyset)$ is a $k$-fair instance and let $\mu = \mu_k(G, w, N, \emptyset)$. By Lemma~\ref{max measure size} we have that $\mu < 1050k^2 \cdot N \cdot \log^2(N) = n^{O(1)}$. Hence, by Lemma~\ref{T asymptotics} it follows that $T(G, w, N, \emptyset) \leq T(\mu) = \mu^{O(k^2 \cdot \log^3(\mu))} = n^{O(k^2 \cdot \log^3(n))}$. 
\end{proof}

\section{Disconnected Forbidden Induced Subgraphs}\label{sec:disconnected}
Let $H$ be a graph. We denote by O$_H$ an oracle that takes an $H$-free graph $G$ as input and outputs the weight of a maximum weight independent set in $G$. In this section we present a quasi-polynomial time algorithm for {\sc Maximum Weight Independent Set} in $H$-free graphs, assuming we have access to the oracles O$_C$ for all $C \in {\cal CC}(H)$. Specifically we will prove Theorem~\ref{thm:main 2}. 

In the following, $H$ = $H_0$ + $H_1$ +...+ $H_{c-1}$ is a graph, $G$ is a graph, $w$ is a weight function on the vertices of $G$, $N$ is a positive integer, and ${\cal F}$ is a vertex multi-family of subsets of $V(G)$. We now present the algorithm $\mbox{ALG}_2$ of Theorem~\ref{thm:main 2}. The algorithm is very similar to the algorithm $\mbox{ALG}_1$ for $P_k$ free graphs, the main difference is that instead of packing balanced separators in the family ${\cal F}$, the algorithm ``packs'' (neighborhoods of) copies of induced $H_i$'s. 

\medskip
\noindent \textbf{ALG$_2$}
\begin{algorithmic}[1]
\STATE \textbf{input:} $H, G, w, N, {\cal F}$.
\STATE \textbf{output:} {\sf mwis}$(G)$.
\STATE $i = |{\cal F}| \mod c$
\IF{exists branchable vertex $v$}
\RETURN $\max\left(\mbox{ALG}_2(H, G - v, w, N, {\cal F} - \{v\}), \mbox{ALG}_2(H, G - N[v], w, N, {\cal F} - N[v]) + w(v) \right)$
\ELSIF{exists induced $H_i$}
\STATE \textbf{obtain} $X \gets$ induced $H_i$ in $G$ 
\RETURN ALG$_2(H, G, w, N, {\cal F} \cup \{N[X]\})$
\ENDIF
\RETURN $O_{H_i}(G)$
\end{algorithmic}
\medskip

The proof of correctness and running time analysis for $\mbox{ALG}_2$ closely follows that of $\mbox{ALG}_1$. The main difference is in the proof of why the family ${\cal F}$ can not grow beyond size $\log N$ (Lemmata~\ref{claw sequence} and~\ref{F size 2}). The other parts are just minor modifications of corresponding results from Section~\ref{sec:p5alg}.

We will distinguish between the two different kinds of recursive calls that ALG$_2$ can make. If the {\bf if} condition of line 4 holds, then the algorithm makes the recursive calls on line 5. In this case we say that ALG$_2$ {\em branches on a branchable vertex}. If the {\bf else if} condition of line 6 holds, then the algorithm makes the recursive call in line 8. In this case we say that ALG$_2$ {\em adds a neighborhood}. We define instances, runs, calls, execution and making a recursive call similarly as for $\mbox{ALG}_1$.
Just as for $\mbox{ALG}_1$, a run of $\mbox{ALG}_2(H, G,w,N,{\cal F})$ is called a $fair$ $run$ if $G$ is an $H$-free graph, $N = |V(G)|$, ${\cal F} = \emptyset$, and $w$ is a weight function. A call $\mbox{ALG}_2(H, G, w, N, {\cal F})$ is called a $fair$ $call$ if it is executed during the course of a fair run. An instance $(H, G, w, N, {\cal F})$ is a {\em fair instance} if $\mbox{ALG}_2(H, G, w, N, {\cal F})$ is a fair call.

\begin{lemma}\label{termination 2}
$ALG_2(H, G, w, N, {\cal F})$ terminates on every input.
\end{lemma}

\begin{proof}
Consider a run of ALG$_2$ with initial input $H, G, w, N,$ and ${\cal F}$. Whenever the algorithm makes a recursive call it does so with a call ALG$_2(H, G', w, N, {\cal F}')$ where $|V(G')| \leq |V(G)|$. Furthermore, whenever the algorithm branches on a branchable vertex, then it recurses with a call ALG$_2(H, G', w, N,{\cal F})$ where $|V(G')| < |V(G)|$. Furthermore ALG$_2$ can not add a neighborhood in over $|V(G)| \cdot \log(N)$ successive recursive calls. Suppose it does, then a call ALG$_2(H, G, w, N, {\cal F}'')$ with ${\cal F''} = |V(G)| \cdot \log(N)$ adds a neighborhood, but $L({\cal F''}, \log(N)) \neq \emptyset$ and thus there exists a branchable vertex, contradicting that ALG$_2(H, G, w, N, {\cal F}'')$ with ${\cal F''} = |V(G)| \cdot \log(N)$ adds a neighborhood. It follows by induction on $|V(G)|$ that ALG$_2$ always terminates.
\end{proof}

\begin{lemma}\label{mis return 2}
A run ALG$_2(H, G, w, N, {\cal F})$ returns the weight of a maximum weight independent set of $G$.
\end{lemma}

\begin{proof}
Consider an run of ALG$_2$ with initial input $(H, G, w, N,{\cal F})$. It is clear from the algorithm that if each run ALG$_2$($H, G', w, N, {\cal F}'$) that is executed by the call ALG$_2$($H, G, w, N, {\cal F}$) returns the weight of a maximum weight independent set of $G'$ under the weight function $w$, then so would the run ALG$_2(H, G, w, N, {\cal F}$). By Lemma \ref{termination 2} the height of the recursion tree is bounded, and the result is trivially true for the base case of $|V(G)| \leq 1$ so the result follows by induction on the depth of the recursion tree.
\end{proof}



\begin{observation}\label{log(N) + 1 is empty 2}
For ever fair instance $(H, G, w, N, {\cal F})$, we have that $L({\cal F},\log(N)+1) = \emptyset$. 
\end{observation}

\begin{proof}
Consider a fair call ALG$_2(H, G, w, N, {\cal F})$. We will prove the result by induction on the depth of the call ALG$_2(H, G, w, N, {\cal F})$ in the recursion tree of a run ALG$_2(H, G^*,w,|V(G^*)|,\emptyset)$ which executes  ALG$_2(H, G, w, N, {\cal F})$.

If ${\cal F} = \emptyset$ then the result is trivially true. Suppose ${\cal F} \neq \emptyset$, it follows that ALG$_2$ executes ALG$_2(H, G, w, N, {\cal F})$ during a fair call ALG$_2(H, G', w, N, {\cal F'})$ by branching on a branchable vertex or by adding a neighborhood, $N[X]$. In the first case, it is clear that since ${\cal F}$ =  ${\cal F'}-S$ for some vertex set $S$, if $L({\cal F'},\log(N)+1) = \emptyset$, then $L({\cal F},\log(N)+1) = \emptyset$ in ALG$_2(H, G, w, N, {\cal F})$. In the second case, since the instance ALG$_2(H, G, w, N, {\cal F'})$ does not branch on a branchable vertex, we have that $L({\cal F'},\log(N)) = \emptyset$ since every vertex in $L({\cal F'},\log(N))$ is branchable. It follows that $L({\cal F},\log(N)+1) = L({\cal F'} \cup \{N[X]\},\log(N)+1) = \emptyset$.
\end{proof}

\begin{lemma}\label{claw sequence}
Let $G$ be a graph, $N$ an integer greater than 1, and let $H = H_0 + H_1 +...+ H_{c-1}$ be a graph. If there exists a sequence of subsets of $V(G)$, $\{X_m\}$ =  $X_0, X_1,..., X_{c \cdot |H| \cdot \log(N)-1}$ such that for all $i$, $X_i \subset V(G)$, the subgraph induced by $X_i$ is isomorphic to $H_{i\ (mod\ c)}$, and for all $v \in X_i$ we have that $\{v\} \cap N[X_j] \neq \emptyset$ for at most $\log(N)$ $X_j$'s where $j < i$, then there exists a subset $I \subseteq \{0,1,2,\ldots, c \cdot |H| \cdot \log(N)-1\}$ such that $X_I = \bigcup_{i \in I} X_i$ forms an induced $H$ in $G$.
\end{lemma}

\begin{proof}
Let $G$ and $H$ be graphs, $N$ an integer greater than 1, and $X_0, X_1,..., X_{c \cdot |H| \cdot \log(N)-1}$ a sequence of sets of vertices with the properties given in the statement of the lemma. Given an $X_j$, set $i$ = $j - (j$ $(mod$ $c)$. We will refer to the segment $X_i$, $X_{i+1}$,..., $X_{i+c-1}$ as $X_j$'s block. 

The proof is by induction on $c$. If $c = 1$ then the statement is trivially true. Assume now that $c > 1$ and that the statement is true for all smaller values. There are at most $|H_{c-1}| \cdot \log(N)$ $X_j$'s such that some vertex of $X_{c \cdot |H| \cdot \log(N)-1}$ belongs to $X_j$, $j \neq c \cdot |H| \cdot \log(N)-1$. Remove from the sequence each such $X_j$ along with all other vertex sets in $X_j$'s block, as well as all $X_t$'s such that $c-1 \equiv t\ (mod\ c)$. After these deletions, re-name the sets $X_j$ in the updated sequence so that the index $j$ of each set $X_j$ is equal to the position of $X_j$ in the sequence (starting with $X_0$). 

Let $H' = H - H_{c-1}$. There are at least $\log(N) \cdot (c \cdot |H| - c \cdot |H_{c-1}| - |H|  + |H_{c-1}|) - 1$ = $\log(N) \cdot (c-1) \cdot |H'| - 1$ remaining vertex sets in the updated sequence, and this new sequence along with $H'$ and $G$ satisfies the condition of the inductive hypothesis. It follows that there exists a set $X_I'$ such that $G[X_I'] = H'$ and $X_I'$ is the union of sets in the (updated) sequence. Since $X_{c \cdot |H| \cdot \log(N)-1}$ does not belong to the neighborhood of any of the vertex sets in the new sequence, $X_{c \cdot |H| \cdot \log(N)-1}$ is disjoint from $N[X_I']$, and hence $X_I = X_I' \cup X_{c \cdot |H| \cdot \log(N)-1}$ induces $H$ in $G$, completing the proof.
\end{proof}

\begin{lemma}\label{F size 2}
For every fair instance $(H, G, w, N, {\cal F})$ with $H = H_0 + H_1 +...+ H_{c-1}$, it holds that $|{\cal F}| < c \cdot |H| \cdot \log(N)$
\end{lemma}

\begin{proof}
Let the fair instance $(H, G, w, N, {\cal F})$ be as in the statement of the lemma, furthermore let $G'$ be the graph used in the initial input of ALG$_2$ of the fair run that produces the instance $(H, G, w, N, {\cal F})$. Assume to the contrary, that $|{\cal F}| \geq c \cdot |H| \cdot \log(N)$. In the fair run that executes the call ALG$_2(H, G, w, N, {\cal F})$, consider the sequence of recursive calls (ordered by when the call occurs) that lead to the call ALG$_2(H, G, w, N, {\cal F})$. In particular, consider the subsequence 
$$\mbox{ALG}^0_2(H, G^0, w, N, {\cal F}^0),  \mbox{ALG}^1_2(H, G^1, w, N, {\cal F}^1),..., \mbox{ALG}^{c \cdot |H| \cdot \log(N)-1}_2(H, G^{c \cdot |H| \cdot \log(N)-1}, w, N, {\cal F}^{c \cdot |H| \cdot \log(N)-1})$$
such that the call ALG$^i_2(H, G^i, w, N, {\cal F}^i)$ is the $(i+1)^{th}$ call to add a neighborhood $N[X_i]$. By Observation \ref{log(N) + 1 is empty 2}, we can see that for all $X_i$, and for all vertices $v \in X_i$, $\{v\} \cap N[X_j] \neq \emptyset$ for at most $\log(N)\ X_j$'s with $j < i$. The result follows now by observing that $G'$, $H$, $N$, and the sequence $X_0, X_1,..., X_{c \cdot |H| \cdot \log(N)-1}$ satisfy the hypothesis of Lemma \ref{claw sequence}, contradicting that $G'$ is $H$-free.



\end{proof}

\begin{observation}\label{level set size 2}
For every fair call ALG$_2(H, G, w, N, {\cal F})$ that recurses by adding a neighborhood $N[X]$ and for every i, $$|L({\cal F}\cup N[X],i)| \leq \Delta_{i-1} \cdot |H| + |L({\cal F},i)|$$ Furthermore, for every fair instance $(H, G', w, N', {\cal F}')$, $$|L({\cal F}',i)| \leq \Delta_{i-1} \cdot |H| \cdot |{\cal F}'|$$
\end{observation}

\begin{proof}
Consider a fair call ALG$_2(H, G, w, N, {\cal F})$ that recurses by adding a neighborhood $N[X]$. Let $X_j$ denote the set of vertices in $L({\cal F},j) \cap N[X]$, then we can see that $|L({\cal F}\cup \{N[X]\},j)| \leq L({\cal F},j) + |X_{j-1}|.$ Since the call ALG$_2(H, G, w, N, {\cal F})$ adds a neighborhood, $N[X]$, there are no branchable vertices. So, we have that for all $v \in G$, $|N[v] \cap L({\cal F}, j)| < \Delta_j$. Hence $|X_{j-1}|$ $<$ $\Delta_{j-1} \cdot |H|$ and the result $|L({\cal F}\cup \{N[X]\},i)| \leq \Delta_{i-1} \cdot |H| + |L({\cal F},i)|$ follows.

The second inequality follows by combining induction, the first part the observation, and the fact that if the call ALG$_2(H, G, w, N, {\cal F})$ executes ALG$_2(H, G', w, N', {\cal F}')$, then $|{\cal F}| < |{\cal F}'|$ if and only if the call ALG$_2(H, G, w, N, {\cal F})$ makes the call ALG$_2(H, G', w, N', {\cal F}')$ by adding a neighborhood.
\end{proof}

For fair instances $(H, G, w, N,{\cal F})$ we define the measure
\begin{equation*}
 \label{eq2}
\begin{split}
\mu_H(H, G, w, N,{\cal F})  &= |V(G)| + \sum_i\left(|L({\cal F},i)| \cdot \frac{N}{\Delta_{i-1}}\right)\\
& \quad + 2|H| \cdot N \cdot \log(N) \cdot (|H| \cdot |{\cal CC}(H)| \cdot \log(N)-|{\cal F}|)\\
\end{split}
\end{equation*}
If $(H, G, w, N,{\cal F})$ is not a fair instance, then $\mu_H(H, G, w, N,{\cal F})$ is undefined. Note that $\mu_H(H, G, w, N,{\cal F})$ must always be an integer and that it is independent of the weight function $w$. We will say that two instances $(H, G, w, N,{\cal F})$ and $(H, G', w', N',{\cal F}')$ are {\em essentially different} if $G' \neq G$, $N' \neq N$ or ${\cal F}' \neq {\cal F}$. 

If $N = 1$ then a fair run ALG$_2(H, G, w, N,{\cal F})$ clearly terminates after a constant number of steps (since in a fair run, $|V(G)| \leq N$) regardless of the other inputs, so from now on we will assume $N > 1$.

\begin{lemma}\label{finite measure 2}
For every positive integer $\mu$, the number of essentially different fair instances $(H, G, w, N,{\cal F})$ such that  $\mu_H(H, G, w, N,{\cal F})$ = $\mu$ is finite. In addition, for every fair instance $\mu(H, G, w, N,{\cal F})$ $\geq$ 0.
\end{lemma}

\begin{proof}
Consider a fair instance $(H, G, w, N, {\cal F})$. We will show that if $\mu_H(H, G, w, N,{\cal F})$ = $\mu$, then $|V(G)| \leq \mu$. If $|V(G)| \leq \mu$ then the range of inputs for $G$, $N$, and ${\cal F}$ are bounded in terms of $\mu$ and the first part of the Lemma follows.

By Lemma \ref{F size 2} we have that $|{\cal F}|$ is less than $|H| \cdot |{\cal CC}(H)| \cdot \log(N)$. It follows that the terms $|V(G)|$, $\Sigma_i(|L({\cal F},i)| \cdot \frac{N}{\Delta_{i-1}})$, and $2|H| \cdot N \cdot \log(N) \cdot (|H| \cdot |{\cal CC}(H)| \cdot \log(N)-|{\cal F}|)$ are all non negative. Hence $\mu(H, G, w, N,{\cal F})$ $\geq$ $|V(G)|$. This also proves that $\mu_H(H, G, w, N,{\cal F})$ $\geq$ 0.
\end{proof}

\begin{lemma}\label{max measure size 2}
$\mu_H(H, G, w, N,{\cal F})$ $\leq$ $4|H|^2 \cdot |{\cal CC}(H)| \cdot N \cdot \log^2(N)$ for every fair instance $(H, G, w, N,{\cal F})$.
\end{lemma}

\begin{proof}
Consider a fair instance $(H, G, w, N,{\cal F})$. By Observation \ref{level set size 2} and Lemma \ref{F size 2}, we have that  
$$|L({\cal F},i)| \cdot \frac{N}{\Delta_{i-1}} < |H| \cdot N \cdot |{\cal CC}(H)| < |H|^2 \cdot |{\cal CC}(H)| \cdot N \cdot \log(N)$$ 
It follows that 
$$\sum_i\left(|L({\cal F},i)| \cdot \frac{N}{\Delta_{i-1}}\right) < |H|^2 \cdot |{\cal CC}(H)| \cdot N \cdot \log^2(N)$$
Also, since $N \geq |V(G)|$, we have the following. 
%
%
\begin{equation*}
 \label{eq2}
\begin{split}
\mu_H(H, G, w, N,{\cal F}) & = |V(G)| + \Sigma_i(|L({\cal F},i)| \cdot \frac{N}{\Delta_{i-1}}) \\
& \quad + 2|H| \cdot N \cdot \log(N) \cdot (|H| \cdot |{\cal CC}(H)| \cdot \log(N)-|{\cal F}|) \\
& < |V(G)| + |H|^2 \cdot |{\cal CC}(H)| \cdot N \cdot \log^2(N) + 2|H|^2 \cdot |{\cal CC}(H)| \cdot N \cdot \log^2(N)\\ 
& = 4|H|^2 \cdot |{\cal CC}(H)| \cdot N \cdot \log^2(N)
\end{split}
\end{equation*}

\end{proof}

We define $T_H(H, G, w, N, {\cal F})$ to be the running time (including the {\em number} of oracle calls) of ALG$_2$ starting with the inputs $(H, G, w, N, {\cal F})$. We also define 
$$T_H(\mu) = \max_{\substack{G, N,{\cal F} \mbox{ s.t.} \\ \mu_H(H, G, w, N,{\cal F}) \leq \mu}} \mbox{T}_H(H, G, w, N,{\cal F})$$ 

Just as for $\mbox{ALG}_1$, when we analyze run time we assume that arithmetic on weights takes constant time. Thus, both the running time of $ALG_2$ and the measure of an instance  $(H, G, w, N, {\cal F})$ are independent of the weight function $w$, and so by Lemma \ref{finite measure 2}, $T_H(\mu)$ is well defined.

\begin{lemma}\label{T recurrence 2}
$T_H(\mu)$ satisfies the following recurrence:

\begin{equation*}
T_H(\mu) \leq \mu^{O(1)} + max
\begin{cases}
T_H(\mu - 1) + T_H(\mu[1-\frac{1}{8|H|^2 \cdot |{\cal CC}(H)| \cdot \log^2(\mu)}])\\
T_H(\mu[1-\frac{1}{4|H| \cdot \log(\mu)}])\\
\end{cases}
\end{equation*}

\end{lemma}

\begin{proof}
Let $(H, G, w, N,{\cal F})$ be a fair instance such that $\mu_H(H, G, w, N,{\cal F}))$ = $\mu$ and $T_H(\mu)$ is the run time of ALG$_2(H, G, w, N,{\cal F})$. If the call ALG$_2(H, G, w, N,{\cal F})$ branches on a branchable vertex, $v$, then it makes two recursive calls, one execution on $(H, G-\{v\}, w, N, {\cal F}-\{v\})$, which has measure at most $\mu - 1$. The other execution is on the instance $(H, G-N[v], w, N, {\cal F}-N[v])$. Note that for a branchable vertex, $v$, we have that 
$$\sum_i\left(|L({\cal F}-N[v],i)| \cdot \frac{N}{\Delta_{i-1}}\right) \leq \sum_i\left(|L({\cal F},i)| \cdot \frac{N}{\Delta_{i-1}}\right) - \frac{N}{2}$$ 

since for at least one level $i$ we have that $|N[v] \cap L({\cal F},i)| \geq \Delta_i$ and $\frac{\Delta_i}{\Delta_{i-1}}$ = 1/2.

Hence,


\begin{equation*}
 \label{eq1}
\begin{split}
\mu_H(H, G-N[v], w, N,{\cal F}-N[v]) & = |V(G)-N[v]| + \Sigma_i(|L({\cal F}-N[v],i)| \cdot \frac{N}{\Delta_{i-1}})\\
&  \quad + 2|H| \cdot N \cdot \log(N) \cdot (|H| \cdot |{\cal CC}(H)| \cdot \log(N)-|{\cal F}-N[v]|) \\
& \leq |V(G)| + \sum_i\left(|L({\cal F},i)| \cdot \frac{N}{\Delta_{i-1}}\right)\\
& \quad + 2|H| \cdot N \cdot \log(N) \cdot (|H| \cdot |{\cal CC}(H)| \cdot \log(N)-|{\cal F}|) - \frac{N}{2} \\
& = \mu - \frac{N}{2}\\ 
& \leq \mu\left(1-\frac{1}{8|H|^2 \cdot |{\cal CC}(H)| \cdot \log^2(N)}\right)  \mbox{~~~~~~( by\ Lemma\ \ref{max measure size 2} )}\\
& \leq \mu\left(1-\frac{1}{8|H|^2 \cdot |{\cal CC}(H)| \cdot \log^2(\mu)}\right)
\end{split}
\end{equation*}
Thus, if the call ALG$_2(H, G, w, N,{\cal F})$ branches on a branchable vertex, then we have that $$T_H(\mu) \leq T_H(\mu - 1) + T_H\left(\mu[1-\frac{1}{8|H|^2 \cdot |{\cal CC}(H)| \cdot \log^2(\mu)}]\right)$$

If ALG$_2(H, G, w, N,{\cal F})$ adds a neighborhood, $N[X]$, it makes a single call ALG$_2(H,G, w, N,{\cal F}\cup \{N[X]\})$. By Observation \ref{level set size 2} and Lemma \ref{max measure size 2} we get the following.
$$\mu(H, G, w, N,{\cal F}\cup \{N[X]\}) < \mu + |H| \cdot N \cdot \log(n) - 2|H| \cdot N \cdot \log(N) < \mu\left([1-\frac{1}{4|H| \cdot |{\cal CC}(H)| \cdot \log(\mu)}]\right)$$
Thus, if the call ALG$_2(H, G, w, N,{\cal F})$ adds a neighborhood, then $T_H(\mu)$ $\leq$ $T_H(\mu([1-\frac{1}{4|H| \cdot |{\cal CC}(H)| \cdot \log(\mu)}])$.
The result now follows from the observation that ALG$_2(H, G, w, N,{\cal F})$ only does $|V(G)|^{O(1)}$ = $\mu^{O(1)}$ work in a given call and always branches on a branchable vertex, adds a balanced separator, or immediately returns a value without making further recursive calls.
\end{proof}

Since $T_H(\mu)$ is a non negative, non decreasing function, by adding the two possibilities in the $\max$ of  Lemma \ref{T recurrence 2} we immediately obtain the following simplified recurrence. 

\begin{corollary}\label{T bound 2}
$T_H(\mu) \leq \mu^{O(1)} + T_H(\mu[1-\frac{1}{4|H| \cdot \log(\mu)}]) +  T_H(\mu - 1)  +  T_H(\mu[1-\frac{1}{8|H|^2 \cdot |{\cal CC}(H)| \cdot \log^2(\mu)}])$ $<$ $T_H(\mu - 1) + \mu^{O(1)} + 2T_H(\mu[1-\frac{1}{8|H|^2 \cdot |{\cal CC}(H)| \cdot \log^2(\mu)}])$
\end{corollary}

\begin{lemma}\label{T asymptotics 2}
$T_H(\mu) = \mu^{O(|H|^2 \cdot |{\cal CC}(H)| \cdot \log^3(\mu))}$
\end{lemma}

\begin{proof}
The proof is by induction on $\mu$. The base case is established by Lemma \ref{finite measure 2}. By Corollary \ref{T bound 2} we have the inequality $T_H(\mu) \leq T_H(\mu - 1) +  \mu^{O(1)} + 2T_H(\mu[1-\frac{1}{8|H|^2 \cdot |{\cal CC}(H)| \cdot \log^2(\mu)}])$ and repeatedly applying the inequality to the first term on the right hand side, gives $T(\mu) \leq \mu^{O(1)} + 2\mu T_H(\mu[1-\frac{1}{8|H|^2 \cdot |{\cal CC}(H)| \cdot \log^2(\mu)}])$. By the inductive hypothesis then, there is some constant $c$ such that $T_H(\mu)$

\begin{equation*}
 \label{eq1}
\begin{split}
 & \leq \mu^{O(1)} + 2\mu(\mu[1-\frac{1}{8|H|^2 \cdot |{\cal CC}(H)| \cdot \log^2(\mu)}])^{c|H|^2 \cdot |{\cal CC}(H)| \cdot \log^3(\mu)}\\
&  = \mu^{O(1)} + 2\mu \mu^{c|H|^2 \cdot |{\cal CC}(H)| \cdot \log^3(\mu)}\left(1-\frac{1}{8|H|^2 \cdot |{\cal CC}(H)| \cdot \log^2(\mu)}\right)^{c|H|^2 \cdot |{\cal CC}(H)| \cdot \log^3(\mu)}\\
& \leq \mu^{O(1)} + 2\mu\mu^{c|H|^2 \cdot |{\cal CC}(H)| \cdot \log^3(\mu)} \cdot e^{-\frac{c|H|^2 \cdot |{\cal CC}(H)| \cdot \log^3(\mu)}{8|H|^2 \cdot |{\cal CC}(H)| \cdot \log^2(\mu)}} \mbox{~~~~~~( since $(1-x) \leq e^{-x}$ )}\\
& \leq  \mu^{O(1)} + 2\mu \mu^{c|H|^2 \cdot |{\cal CC}(H)| \cdot \log^3(\mu)} \cdot e^{-\frac{c\log(\mu)}{8}} \\
& \leq \mu^{c|H|^2 \cdot |{\cal CC}(H)| \cdot \log^3(\mu)} \mbox{~~~~~~( for sufficiently large $c$ )}\\ 
\end{split}
\end{equation*}
\end{proof}

We are now ready to prove Theorem \ref{thm:main 2}.

\begin{proof}[Proof of Theorem~\ref{thm:main 2}]
The algorithm returns the answer of ALG$_2(H,G,w,|V(G)|,\emptyset)$. By Lemmata \ref{termination 2} and~\ref{mis return 2}, ALG$_2$ will always terminate and return the weight of a maximum weight independent set in $G$. 
For the running time analysis, observe that $(H,G,w,|V(G)|,\emptyset)$ is a fair instance and let $\mu = \mu_H(H, G, w, N, {\cal F})$. We assume that $|H| \leq N$, since the run time bound follows trivially if  $|H| > N$. By Lemma \ref{max measure size 2} we have that $\mu < 4|H|^2 \cdot |{\cal CC}(H)| \cdot N \cdot \log^2(N)$. Let $n = N = |V(G)|$, then it follows that 
$$T_H(H, G, w, N, {\cal F}) \leq T_H(\mu) = \mu^{O(|H|^2 \cdot |{\cal CC}(H)| \cdot \log^3(\mu))} = n^{O(|H|^2 \cdot |{\cal CC}(H)| \cdot \log^3(n))}$$
This completes the proof. 
\end{proof}

Theorem~\ref{thm:main 2} sligthly increases the current reach of Theorem~\ref{thm:main}. In particular, let $T_k$ be the graph with $k$ connected components the first of which is a path $P_k$ on $k$ vertices and the remaining $k-1$ are forks (a fork is a path on four vertices plus a single vertex adjacent to the second vertex of the path). Lozin and Milanic~\cite{LozinM06} gave a polynomial time algorithm for {\sc Weighted Independent Set} on fork-free graphs. Theorem~\ref{thm:main 2} implies that {\sc Weighted Independent Set} on $T_k$ free graphs can be solved by making $n^{O(k^3\log^3 (n))}$ oracle calls to the polynomial time algorithm of Lozin and Milanic~\cite{LozinM06} or the algorithm of Theorem~\ref{thm:main}. Thus we obtain the following result.

\begin{theorem}\label{thm:tkfree}
There exists an algorithm that given a $T_k$-free graph $G$ and weight function $w : V(G) \rightarrow \mathbb{N}$, runs in $n^{O(k^3 \log^3 n)}$ time, and outputs the weight of a maximum weight independent set of $G$.
\end{theorem}

\section{Conclusion}\label{sec:conclusion}
In this paper we gave a quasipolynomial time algorithm for {\sc Weighted Independent Set} on $P_k$-free graphs for all integers $k$. The dependence on $k$ in the exponent is $O(k^2)$ and so our algorithm is quasi-polynomial even for $k = \log^{O(1)}n$ and sub-exponential for $k = n^{\frac{1}{2} - \epsilon}$ for $\epsilon > 0$. In light of our algorithm it is tempting to conjecture that {\sc (Weighted) Independent Set} on $P_k$-free graphs can be solved in polynomial time for every $k$. Given how dependent our algorithms are on branching on high degree vertices it looks unlikely that our techniques can lead to polynomial time algorithms for $P_k$-free graphs. Nevertheless it may be possible to extract structural insights from our algorithms that could eventually lead to polynomial time algorithms. 

Our second main result (Theorem~\ref{thm:main 2}) implies that if there exists a quasi-polynomial time algorithm for $H$-free graphs for every subdivided claw $H$ then there exists a quasi-polynomial time algorithm for every finite family ${\cal H}$ such that {\sf NP}-completeness of {\sc Independent Set} on ${\cal H}$-free graphs does not follow from Alekseev's result~\cite{Alekseev82}. Thus, a quasi-polynomial time algorithm for subdivided-claw-free graphs would complete a dichotomy for the complexity of {\sc Independent Set} on ${\cal H}$-free graphs for every finite family $H$: every case is either quasi-polynomial time solvable or {\sf NP}-complete.

\bibliographystyle{alpha}
\bibliography{bibliography}
\end{document}